\documentclass[10pt,conference]{IEEEtran}
\usepackage{graphicx,subfigure,psfrag,amsmath,cases}
\usepackage{latexsym,amsmath,amssymb,subfigure,algorithm,mathtools}
\usepackage{algorithmic}
\usepackage{color}
\usepackage{url}
\usepackage{scrtime}
\usepackage{cite}
\usepackage{subfigure}
\usepackage{bbding}
\usepackage{multicol}

\author{{Zhiqiang Wei, Lou Zhao, Jiajia Guo, Derrick Wing Kwan Ng, and Jinhong Yuan \vspace{-10mm}}
\thanks{The authors are with the School of Electrical
Engineering and Telecommunications, UNSW, Australia (email: zhiqiang.wei@student.unsw.edu.au; lou.zhao@student.unsw.edu.au; jiajia.guo@student.unsw.edu.au; w.k.ng@unsw.edu.au; j.yuan@unsw.edu.au).
Derrick Wing
Kwan Ng is supported under Australian Research Councils (ARC) Discovery Early
Career Researcher Award funding scheme DE170100137.
Jinhong Yuan is currently on leave from UNSW and with CAS.
This work was supported in part by the ARC Discovery Project DP160104566, Linkage Project LP160100708, and CAS Pioneer Hundred Talents Program.}}

\title{\vspace{-7mm}A Multi-Beam NOMA Framework for Hybrid mmWave Systems}

\newtheorem{Thm}{Theorem}
\newtheorem{proof}{proof}

\newtheorem{Cor}{Corollary}

\newtheorem{T-Prob}{Transformed Problem}

\begin{document}
\maketitle
\begin{abstract}
In this paper, we propose a multi-beam non-orthogonal multiple access (NOMA) framework for hybrid millimeter wave (mmWave) systems.
The proposed framework enables the use of a limited number of radio frequency (RF) chains in hybrid mmWave systems to accommodate multiple users with various angles of departures (AODs).
A beam splitting technique is introduced to generate multiple analog beams to facilitate NOMA transmission.
We analyze the performance of a system when there are sufficient numbers of antennas driven by a single RF chain at each transceiver.
Furthermore, we derive the sufficient and necessary conditions of antenna allocation, which guarantees that the proposed multi-beam NOMA scheme outperforms the conventional time division multiple access (TDMA) scheme in terms of system sum-rate.
The numerical results confirm the accuracy of the developed analysis and unveil the performance gain achieved by the proposed multi-beam NOMA scheme over the single-beam NOMA scheme.
\end{abstract}
\vspace{-2mm}
\section{Introduction}
Millimeter wave (mmWave) communications are recognized as one of the most important technologies to achieve tens of gigabit data rate in the fifth-generation (5G) and its applications in outdoor systems have been extensively studied in the literature\cite{Rappaport2013,XiaoMing2017,wong2017key}.
Recently, hybrid mmWave systems \cite{zhao2017multiuser,GaoSubarray,lin2016energy}, where only few radio frequency (RF) chains are deployed to drive a large antenna array, have been proposed for practical implementations of mmWave technology.
Specifically, most of existing works \cite{zhao2017multiuser,GaoSubarray,lin2016energy} have focused on the channel estimation and hybrid precoding design, while the design of potential and efficient multiple access schemes for hybrid mmWave systems is rarely discussed.

Multiple access technology is fundamentally important to support multiuser communications in wireless networks.
Although it has been widely investigated in communication systems utilizing microwave band, it is still an open problem for hybrid mmWave communication systems.
In fact, conventional orthogonal multiple access (OMA) schemes adopted in previous generations of wireless networks cannot be applied directly to the hybrid mmWave systems, due to the special propagation features and constraints of hardware implementations.
Specifically, in hybrid mmWave systems, an analog precoder is shared by all the frequency components of the whole frequency band.
Besides, the beamwidth of an analog beam in mmWave frequency band is typically narrow\footnote{The $3$ dB beamwidth of a uniform linear array (ULA) with $M$ half wavelength spacing antennas is about $\frac{{102.1}}{M}$ degrees \cite{van2002optimum}.}.
Subsequently, the conventional OMA schemes, such as frequency division multiple access (FDMA) and orthogonal frequency division multiple access (OFDMA) are only applicable to the cases where multiple users share the same analog beam.
In other words, the huge bandwidth in mmWave frequency band cannot be utilized efficiently.
Another OMA scheme, time division multiple access (TDMA), might be a good candidate to facilitate multiple user communication in hybrid mmWave systems, where users share the spectrum via orthogonal time slots.
However, it is well-known that the spectral efficiency of TDMA is inferior to that of non-orthogonal multiple access (NOMA) \cite{Ding2015b} and the key challenge of implementing TDMA in hybrid mmWave systems is the requirement of high precision in time synchronization among users since mmWave communications usually provide a high symbol rate.
On the other hand, spatial division multiple access (SDMA) \cite{XiaoMing2017} is a potential technology for supporting multiple users, provided that the base station (BS) is equipped with enough numbers of RF chains and antennas.
However, in practical hybrid mmWave systems, the limited number of RF chains restricts the number of users that can be served simultaneously via SDMA, as one RF chain can serve only at most one user.
Thus, the combination of SDMA and mmWave \cite{zhao2017multiuser} is unable to cope with the emerging need of massive connectivity demanded in the future 5G communication systems\cite{Andrews2014}.
%
%
Therefore, this paper attempts to overcome the limitation of the small number of RF chains in hybrid mmWave systems, i.e., serving more users with a limited number of RF chains, via introducing the concept of NOMA into hybrid mmWave systems.

Recently, NOMA has drawn a significant amount of attention as a promising multiple access technique for 5G\cite{Dai2015,Ding2015b,WeiSurvey2016}.
%
%
In contrast to conventional OMA schemes, e.g. \cite{Kwan_AF_2010,DerrickEEOFDMA}, NOMA can serve multiple users via the same degrees of freedom (DOF) and achieve a higher spectral efficiency\cite{Ding2014}.
The superiority motivates the introduction of NOMA concept into hybrid mmWave systems to accommodate more users with a limited number of RF chains.
Several preliminary works considered NOMA schemes for mmWave communications\cite{Ding2017RandomBeamforming,Cui2017Optimal, WangBeamSpace2017}.
%
However, most of the proposed NOMA schemes in the literature, e.g. \cite{Ding2017RandomBeamforming,Cui2017Optimal, WangBeamSpace2017}, are single-beam NOMA where NOMA transmission is applied to users within the same analog beam.
%
%
Yet, due to the narrow analog beamwidth in hybrid mmWave systems, the single-beam NOMA schemes can only serve multiple users simultaneously when they have similar angle-of-departure (AOD) at the transmitter.
Therefore, the number of users that can be served concurrently by the existing single-beam NOMA schemes is very limited and it depends on the users' AOD distribution.
This will reduce the potential performance gain brought by NOMA in hybrid mmWave systems.
Most recently, the general idea of multi-beam NOMA was proposed and discussed for hybrid mmWave systems in \cite{Xiao2017MultiBeam}, which applies NOMA to multiple users with separated AODs.
However, \cite{Xiao2017MultiBeam} only offered a high-level discussion and the implementation details of multi-beam NOMA in hybrid mmWave systems were not included.
More importantly, the performance gain of the multi-beam NOMA scheme over the single-beam NOMA scheme has not been reported yet.

In this paper, we propose a multi-beam NOMA framework for hybrid mmWave systems, which is more flexible than the conventional single-beam NOMA schemes for serving more users with an arbitrary AODs distribution.
More specifically, all the users are divided into several NOMA groups and each NOMA group is associated with a RF chain.
We generate multiple analog beams for each NOMA group to facilitate downlink NOMA transmission by exploiting the channel sparsity and the large scale antenna array at the BS in hybrid mmWave systems.
Here, the multiple analog beams are generated by the proposed beam splitting technique, which dynamically divides the whole antenna array associated with a RF chain into multiple subarrays.
To provide more insights of the proposed multi-beam NOMA scheme, we analyze the performance of the considered hybrid mmWave system when there are sufficiently large numbers of antennas but only single RF chain equipped at each transceiver.
Furthermore, the sufficient and necessary conditions of antenna allocation are obtained which guarantee the proposed scheme outperforming TDMA with a fixed time and power allocation strategy.
Simulation results verify the developed analytical results and demonstrate that the proposed multi-beam NOMA scheme can achieve a higher system sum-rate than that of the conventional TDMA scheme when the antenna allocation conditions can be satisfied.
%
%
%
%

\vspace{-1mm}
\section{System Model}
\vspace{-1mm}
\subsection{System Model}

\begin{figure}[t]
\centering
{\subfigure[The system model of the proposed multi-beam NOMA scheme for hybrid mmWave systems.]
{\label{MultiBeamNOMA:a} 
\includegraphics[width=0.48\textwidth]{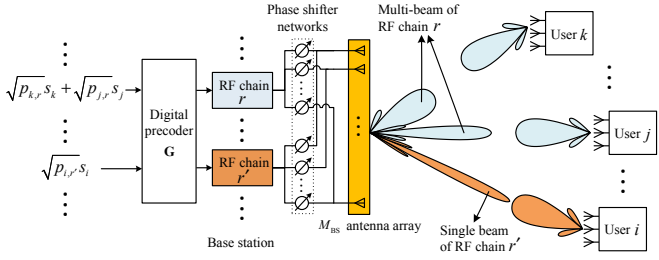}}\vspace{-4mm}}
{\subfigure[A hybrid mmWave structure at user $k$.]
{\label{MultiBeamNOMA:b} 
\includegraphics[width=0.26\textwidth]{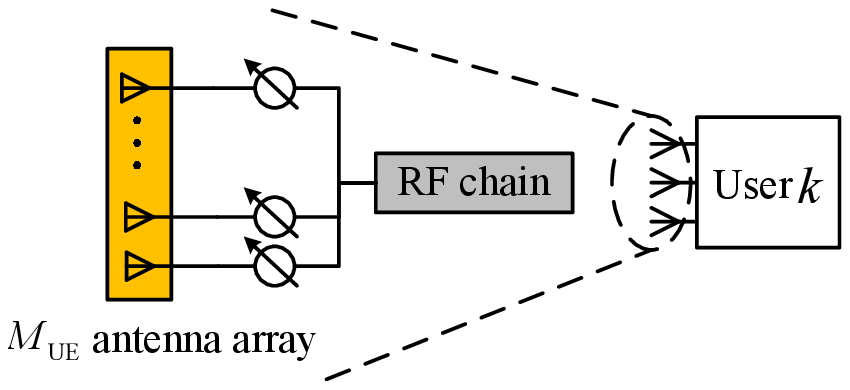}}\vspace{-2mm}}
\caption{System model of the proposed multi-beam NOMA scheme for hybrid mmWave systems.}\vspace{-6mm}
\label{MultiBeamNOMA}%
\end{figure}

We consider a downlink hybrid mmWave communication in a single-cell system with one base station (BS) and $K$ users, as shown in Figure \ref{MultiBeamNOMA}.
We assume that the BS is equipped with $M_{\mathrm{BS}} \ge1$ antennas and there are only $N_{\mathrm{RF}}\ge1$ RF chains such that $M_{\mathrm{BS}} \gg N_{\mathrm{RF}}$.
We note that each RF chain can access all the $M_{\mathrm{BS}}$ antennas through $M_{\mathrm{BS}}$ phase shifters, as shown in Figure \ref{MultiBeamNOMA:a}.
Besides, each user is equipped with $M_{\mathrm{UE}} \ge 1$ antennas connected via a single RF chain, as shown in Figure \ref{MultiBeamNOMA:b}.
We adopt the uniform linear array (ULA) structure as it is commonly used in the literature \cite{zhao2017multiuser}.
We assume that the antennas at each transceiver are deployed and separated with equal-space of half wavelength with respect to the neighboring antennas.
In this work, we focus on the overloaded scenario with $K \ge N_{\mathrm{RF}}$, which is fundamentally different from existing works in hybrid mmWave communications, e.g. \cite{zhao2017multiuser,GaoSubarray,lin2016energy}.
We note that our considered system model is a generalization of the existing works\cite{zhao2017multiuser,GaoSubarray,lin2016energy}.
For example, the considered system is degenerated to the conventional hybrid mmWave systems when $K \le N_{\mathrm{RF}}$ and each NOMA group contains a single user.
\vspace{-1mm}
\subsection{Channel Model}
\vspace{-1mm}
For mmWave communications, we employ the Saleh-Valenzuela channel model \cite{XiaoMing2017} as it has been widely adopted in the literature, and the channel matrix between the BS and user $k$, ${{\mathbf{H}}_k} \in \mathbb{C}^{{ M_{\mathrm{UE}} \times M_{\mathrm{BS}}}}$, can be represented as
\vspace{-2mm}
\begin{equation}\label{ChannelModel1}
{{\mathbf{H}}_k} = {\alpha _{k,0}}{\mathbf{H}}_{k,0} + \sum\nolimits_{l = 1}^L {\alpha _{k,l}{\mathbf{H}}_{k,l}},\vspace{-2mm}
\end{equation}
where ${\mathbf{H}}_{k,0} \in \mathbb{C}^{ M_{\mathrm{UE}} \times M_{\mathrm{BS}} }$ is the line-of-sight (LOS) channel matrix between the BS and user $k$ with ${\alpha _{k,0}}$ denoting the LOS complex path gain.
${\mathbf{H}}_{k,l} \in \mathbb{C}^{ M_{\mathrm{UE}} \times M_{\mathrm{BS}} }$ denotes the $l$-th non-line-of-sight (NLOS) path channel matrix of user $k$ with ${\alpha _{k,l}}$ denoting the corresponding NLOS complex path gains, $1 \le l \le L$, and $L$ denoting the total number of NLOS paths.
In particular, ${\mathbf{H}}_{k,l}$, $\forall l \in \{0,\ldots,L\}$ can be given by
\vspace{-2mm}
\begin{equation}
{\mathbf{H}}_{k,l} = {\mathbf{a}}_{\mathrm{UE}} \left(  \phi _{k,l} \right){\mathbf{a}}_{\mathrm{BS}}^{\mathrm{H}}\left( \theta _{k,l} \right),\vspace{-2mm}
\end{equation}
with
\vspace{-2mm}
\begin{equation}
\hspace{-2mm}{\mathbf{a}}_{\mathrm{BS}}\left( \theta _{k,l} \right) \hspace{-1mm}=\hspace{-1mm} \left[ \hspace{-0.5mm}{{e^{ j \frac{{M_{{\mathrm{BS}}}} \hspace{-0.5mm}-\hspace{-0.5mm} 1}{2} \pi\hspace{-0.5mm}\cos \left( \theta _{k,l} \right)}}, \ldots ,{e^{ - j\frac{{M_{{\mathrm{BS}}}} \hspace{-0.5mm}-\hspace{-0.5mm} 1}{2}\pi \hspace{-0.5mm} \cos \left( \theta _{k,l} \right)}}} \hspace{-0.5mm}\right]^{\mathrm{T}}\vspace{-1.5mm}
\end{equation}
denoting the array response vector \cite{van2002optimum} for the AOD of the $l$-th path ${\theta _{k,l}}$ at the BS and
\vspace{-2mm}
\begin{equation}
\hspace{-2mm}{\mathbf{a}}_{\mathrm{UE}}\left( \phi _{k,l} \right) \hspace{-1mm}=\hspace{-1mm} \left[ {e^{ j \frac{{M_{{\mathrm{UE}}}} \hspace{-0.5mm}-\hspace{-0.5mm} 1}{2}\pi \hspace{-0.5mm}\cos \left( \phi _{k,l} \right)}}, \ldots ,{e^{ - j\frac{{M_{{\mathrm{UE}}}} \hspace{-0.5mm}-\hspace{-0.5mm} 1}{2}\pi \hspace{-0.5mm}\cos \left( \phi _{k,l} \right)}} \right] ^ {\mathrm{T}}\vspace{-1.5mm}
\end{equation}
denoting the array response vector \cite{van2002optimum} for the angle-of-arrival (AOA) of the $l$-th path ${\phi _{k,l}}$ at user $k$.
The operators ${\left( \cdot \right)^{\mathrm{T}}}$ and ${\left( \cdot \right)^{\mathrm{H}}}$ denote the transpose and the Hermitian transpose of a vector or matrix, respectively.
In this work, we adopt the LOS path based effective channel estimation scheme proposed in \cite{zhao2017multiuser}, which has been proved to achieve a considerable performance of fully digital mmWave systems.
Besides, we assume that the partial LOS channel state information, including the AODs ${\theta _{k,0}}$ and the complex path gain ${{\alpha _{k,0}}}$ for all the users, is known at the BS owing to the beam tracking techniques \cite{BeamTracking}.
For a similar reason, the AOA ${\phi _{k,0}}$ is known at user $k$.
\vspace{-1mm}
\section{Multi-beam NOMA Framework}

The block diagram of the proposed multi-beam NOMA framework for hybrid mmWave systems is shown in Figure \ref{MultiBeamNOMAScheme}, which will be detailed in the sequel.
Note that, this paper aims to propose the multi-beam NOMA framework for hybrid mmWave communication systems and demonstrate the performance gain achieved by generating multiple analog beams.
Due to the page limitation, the user grouping, antenna allocation, digital precoder, and power allocation design will be considered in our future work.
\begin{figure}[t]
\vspace{-3mm}
\centering
\includegraphics[width=3in]{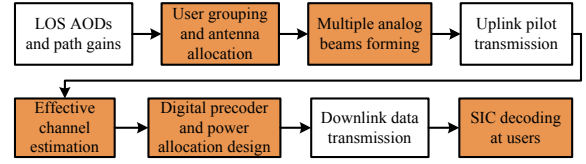}\vspace{-2mm}
\caption{The proposed multi-beam NOMA scheme for hybrid mmWave systems. The shadowed blocks are the design focused of this paper.}\vspace{-6mm}
\label{MultiBeamNOMAScheme}
\end{figure}
\vspace{-1mm}
\subsection{User Grouping and Antenna Allocation}
\vspace{-1mm}
Based on all the LOS AODs of the users $\left\{ {\theta _{1,0}}, \ldots ,{\theta _{K,0}} \right\}$ and their path gains $\left\{ {\alpha _{1,0}}, \ldots ,{\alpha _{K,0}} \right\}$, we perform user grouping and antenna allocation.
In particular, multiple users might be allocated with the same RF chain forming a NOMA group.
We can define the user scheduling variable as follows:
\vspace{-1.5mm}
\begin{equation}
{u_{k,r}} = \left\{ {\begin{array}{*{20}{c}}
1&{{\mathrm{user}}\;k\;{\mathrm{is}}\;{\mathrm{allocated}}\;{\mathrm{to}}\;{\mathrm{RF}}\;{\mathrm{chain}}\;r},\\
0&{{\mathrm{others}}}.
\end{array}} \right.\vspace{-1.5mm}
\end{equation}
To reduce the computational complexity and time delay of SIC decoding within the NOMA group, we restrict at most $G$ users that can be allocated with the same RF chain, i.e., $\sum\nolimits_{k = 1}^K {u_{k,r}}  \le G$, $\forall r \in \{1,\ldots,{N_{\mathrm{RF}}}\}$.
In addition, due to the limited number of RF chains in the considered hybrid systems, we assume that each user can be allocated with at most one RF chain, i.e., $\sum\nolimits_{r = 1}^{N_{\mathrm{RF}}} {u_{k,r}}  \le 1$, $\forall k \in \{1,\ldots,K\}$.

The beam splitting is realized via allocating adjacent antennas to form multiple subarrays and generating an analog beam on each subarray.
For antenna allocation, let $M_{k,r}$ denote the number of antennas allocated to user $k$ associated with RF chain $r$, with $\sum\nolimits_{k = 1}^K u_{k,r}{M_{k,r}}  \le M_{\mathrm{BS}}$, $\forall r$.
\vspace{-1mm}
\subsection{Multiple Analog Beams with Beam Splitting}
\vspace{-1mm}
Now, to generate multiple analog beams with beam splitting, we adopt the coefficients of the analog precoder according to ${M_{k,r}}$ and the LOS AOD of user $k$, ${\theta _{k,0}}$, for each subarray.
For instance, in Figure \ref{MultiBeamNOMA:a}, user $k$ and user $j$ are scheduled to be served by RF chain $r$ at the BS, where their allocated number of antennas are ${M_{k,r}}$ and ${M_{j,r}}$, respectively, satisfying ${M_{k,r}} + {M_{j,r}} \le {M_{\mathrm{BS}}}$.
%
%
Then, the analog precoder for the ${M_{k,r}}$ antennas subarray is given by
\vspace{-2mm}
\begin{align}\label{SubarayWeight1}
&{\mathbf{w}}\left( {M_{k,r}},{\theta _{k,0}} \right)= \notag\\[-1mm]
&{\frac{1}{{\sqrt {{M_{BS}}} }}\left[ {e^{j\frac{ M_{k,r} \hspace{-0.5mm}-\hspace{-0.5mm} 1 }{2}\pi\hspace{-0.5mm} \cos \left( {\theta _{k,0}} \right)}}, \ldots ,{e^{ - j\frac{ M_{k,r} \hspace{-0.5mm}-\hspace{-0.5mm} 1 }{2}\pi\hspace{-0.5mm}\cos \left( {\theta _{k,0}} \right)}} \right] ^ {\mathrm{T}}},
\end{align}
\par
\vspace{-2mm}
\noindent
and the analog precoder for the ${M_{j,r}}$ antennas subarray is given by
\vspace{-2mm}
\begin{align}\label{SubarayWeight2}
&{\mathbf{w}}\left( {M_{j,r},{\theta _{j,0}}} \right) = \notag\\[-1mm]
&{\frac{1}{\sqrt{M_{BS}}}\left[ e^{j\frac{ M_{j,r} \hspace{-0.5mm}-\hspace{-0.5mm} 1 }{2}\pi\hspace{-0.5mm} \cos \left( {{\theta _{j,0}}} \right)}, \ldots ,e^{ - j\frac{ M_{j,r} \hspace{-0.5mm}-\hspace{-0.5mm} 1 }{2}\pi\hspace{-0.5mm} \cos \left( {\theta _{j,0}} \right)} \right] ^ {\mathrm{T}}},
\end{align}
\par
\vspace{-2mm}
\noindent
where ${\mathbf{w}}\left( M_{k,r},{\theta _{k,0}} \right) \in \mathbb{C}^{ M_{k,r} \times 1}$ and ${\mathbf{w}}\left( M_{j,r},\theta _{j,0} \right) \in \mathbb{C}^{ M_{j,r} \times 1}$.
The same normalized factor $\frac{1}{\sqrt {M_{BS}} }$ in \eqref{SubarayWeight1} and \eqref{SubarayWeight2} is introduced to fulfill the constant modulus constraint of phase shifters\cite{XiaoMing2017}.
As a result, the analog precoder for RF chain $r$ is given by
\vspace{-2mm}
\begin{equation}\label{ArrayWeight1}
{\mathbf{w}}_r = \left[ {\begin{array}{*{20}{c}}
{{\mathbf{w}}^{\mathrm{T}}\left( {M_{k,r}},{\theta _{k,0}} \right)}&{{{\mathbf{w}}^{\mathrm{T}}}\left( {M_{j,r}},{\theta _{j,0}} \right)}
\end{array}} \right]^{\mathrm{T}}.\vspace{-2mm}
\end{equation}
On the other hand, if user $i$ is allocated with RF chain $r'$ exclusively, then all the ${M_{\mathrm{BS}}}$ antennas of RF chain $r'$ will be allocated to user $k$.
Besides, the analog precoder for user $i$ is identical to the conventional analog precoder in hybrid mmWave systems, and it is given by
\vspace{-2mm}
\begin{align}\label{ArrayWeight2}
&{{\mathbf{w}}_r'}= {\mathbf{w}}\left( {{M_{\mathrm{BS}}},{\theta _{i,0}}} \right) = \notag\\[-1mm]
&{\frac{1}{{\sqrt {{M_{BS}}} }}\left[ {{e^{j\frac{{ {{M_{\mathrm{BS}}} \hspace{-0.5mm}-\hspace{-0.5mm} 1} }}{2}\pi \hspace{-0.5mm}\cos \left( {{\theta _{i,0}}} \right)}}, \ldots ,{e^{ - j\frac{{ {{M_{\mathrm{BS}}} \hspace{-0.5mm}-\hspace{-0.5mm} 1} }}{2}\pi \hspace{-0.5mm}\cos \left( {{\theta _{i,0}}} \right)}}} \right] ^ {\mathrm{T}}}.
\end{align}
\par

\noindent
Note that compared to single-beam NOMA schemes for mmWave systems\cite{Ding2017RandomBeamforming,Cui2017Optimal, WangBeamSpace2017}, the LOS AODs ${\theta _{k,0}}$ and ${\theta _{j,0}}$ in the proposed scheme are not required to be in the same analog beam.
In other words, our multi-beam NOMA scheme provides a higher flexibility in user grouping.

\begin{figure}[t]
\vspace{-7mm}
\centering
\includegraphics[width=3.0in]{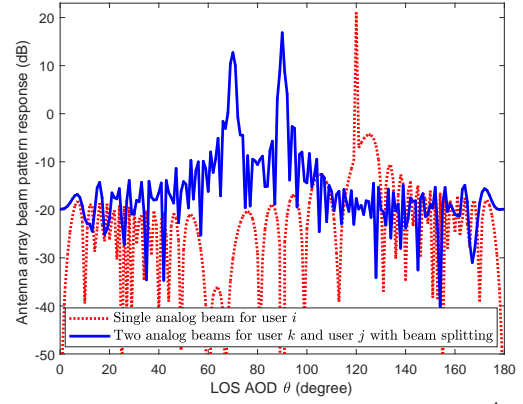}\vspace{-4mm}
\caption{Antenna array beam pattern response for ${{\mathbf{w}}_r}$ in \eqref{ArrayWeight1} and ${{\mathbf{w}}_r'}$ in \eqref{ArrayWeight2} at the BS via the beam splitting technique. We assume ${M_{\mathrm{BS}}} = 128$, ${M_{k,r}} = 50$, ${M_{j,r}} = 78$, ${\theta _k} = 70^{\circ}$, ${\theta _j} = 90^{\circ}$, and ${\theta _i} = 120^{\circ}$.}\vspace{-4mm}
\label{MultiBeamNOMA2}
\end{figure}

Based on the analog precoders ${{\mathbf{w}}_r}$ and ${{\mathbf{w}}_r'}$, RF chain $r$ will generate two analog beams steering toward user $k$ and user $j$, respectively, while RF chain $r'$ will generate a single analog beam steering to user $i$.
The antenna array beam pattern responses for ${{\mathbf{w}}_r}$ and ${{\mathbf{w}}_r'}$ are shown in Figure \ref{MultiBeamNOMA2} to illustrate the multiple analog beams generated via beam splitting.
Compared to the single analog beam for user $i$, we can observe the decreases in the  magnitude of the main beam response  and the increases in beamwidth of the two analog beams for user $k$ and $j$, respectively.
In other words, forming multiple analog beams via beam splitting decreases the beamforming gain which may slightly reduce the beamforming efficiency.
In fact, the flexibility of the proposed scheme in accommodating more users via multiple beams comes at the expense of a lower gain in each beam.
Intuitively, the proposed multi-beam NOMA scheme generates multiple virtual tunnels in the beam-domain for downlink NOMA transmission.
It is worth to note that the beam splitting technique is essentially an array beamforming gain allocation method.
Apart from the power domain multiplexing in conventional NOMA schemes\cite{WeiSurvey2016}, the proposed multi-beam NOMA scheme further exploit the beam-domain for efficient multi-user multiplexing.
%

%
%

Now, to generalize the system model for generating arbitrary numbers of analog beams, we integrate the users scheduling variables ${u_{k,r}}$ with ${{\mathbf{w}}_r}$ as follows
\vspace{-2mm}
\begin{equation}
\hspace{-3mm}{{\mathbf{w}}_r} \hspace{-1mm}=\hspace{-1.2mm} {\left[
{{{\mathbf{w}}^{\mathrm{T}}}\hspace{-1.2mm}\left( {{u_{1,r}},{M_{1,r}},{\theta _{1,0}}} \right)}, \ldots ,{{{\mathbf{w}}^{\mathrm{T}}}\hspace{-1.2mm}\left( {{u_{K,r}},{M_{K,r}},{\theta _{K,0}}} \right)}
\right]^{\mathrm{T}}},\hspace{-1.5mm}\vspace{-2mm}
\end{equation}
with
\vspace{-2mm}
\begin{equation}\label{Weight2}
{\mathbf{w}}\left( {{u_{k,r}},{M_{k,r}},{\theta _{k,0}}} \right) = \left\{ {\begin{array}{*{20}{c}}
\emptyset &{{u_{k,r}} = 0},\\
{\mathbf{w}}\left( {{M_{k,r}},{\theta _{k,0}}} \right) &{{u_{k,r}} = 1}.
\end{array}} \right.\vspace{-2mm}
\end{equation}
It can be observed in \eqref{Weight2} that ${\mathbf{w}}\left( {{u_{k,r}},{M_{k,r}},{\theta _{k,0}}} \right)$ is an empty set $\emptyset$ when ${{u_{k,r}} = 0}$, $\forall k$ and ${{\mathbf{w}}_r}$ only consists of the analog precoders of the users allocated with RF chain $r$, i.e., ${{u_{k,r}} = 1}$, $\forall k$.
So far, the user grouping, antenna allocation, and multiple analog beams forming in the proposed multi-beam NOMA scheme only rely on LOS AODs ${\theta _{k,0}}$ and the complex path gains ${{\alpha _{k,0}}}$, which is fundamentally different from the existing fully digital MIMO-NOMA schemes\cite{Hanif2016}.
\vspace{-1mm}
\subsection{Effective Channel Estimation}
\vspace{-1mm}
%

Now, for a given user grouping strategy, antenna allocation, multiple analog beams forming, all the users transmit their unique orthogonal pilots to the BS in the uplink to perform effective channel estimation.
In this paper, we adopt the time division duplex (TDD) and exploit the channel reciprocity, i.e.,  the estimated effective channel in the uplink can be used for digital precoder design in the downlink.
The effective channel of user $k$ on RF chain $r$ at the BS is given by\cite{zhao2017multiuser}
\vspace{-2mm}
\begin{equation}\label{EffectiveChannel1}
{\tilde h_{k,r}} = {{\mathbf{v}}_k^{\mathrm{H}}}{{\mathbf{H}}_{k}}{\mathbf{w}}_r,\;\forall k,r,\vspace{-2mm}
\end{equation}
where ${{\mathbf{v}}_k} = \frac{1}{{\sqrt {{M_{{\mathrm{UE}}}}} }}{{\mathbf{a}}_{{\mathrm{UE}}}}\left( {{\phi _{k,0}}} \right)$ denotes the analog beamforming vector adopted at user $k$.
We note that, the direction of analog beamforming at each user is always steered toward the BS to fully exploit the beamforming gain.
Using the beam response function of the ULA \cite{van2002optimum}, the effective channel of user $k$ on RF chain $r$ at the BS can be rewritten as
\vspace{-2mm}
\begin{equation}\label{EffectiveChannel2}
{{\tilde h}_{k,r}} \hspace{-1mm}=\hspace{-1mm} \sum\limits_{l = 0}^L {\sum\limits_{k' = 1}^K \hspace{-1mm}{\frac{{{u_{k',r}}{\alpha _{k,l}}}}{{\sqrt {{M_{{\mathrm{UE}}}}{M_{{\mathrm{BS}}}}} }}\frac{{\sin \left( {{M_{{\mathrm{UE}}}}{\kappa _{k,l}}} \right)}}{{\sin \left( {{\kappa _{k,l}}} \right)}}\frac{{\sin \left( {{M_{k',r}}{\varphi _{k,k',l}}} \right)}}{{\sin \left( {{\varphi _{k,k',l}}} \right)}}} },\vspace{-1mm}
\end{equation}
where ${\kappa _{k,l}} = \frac{\pi }{2}\left( {\cos {\phi _{k,0}} - \cos {\phi _{k,l}}} \right)$ and ${\varphi _{k,k',l}} = \frac{{\pi}}{2 }\left( {\cos {\theta _{k,0}} - \cos {\theta _{k',l}}} \right)$ denote the phase differences of received signals from different AODs/AOAs.
Note that, even if user $k$ is not allocated with RF chain $r$, i.e., ${u_{k,r}} = 0$, it still has an effective channel on RF chain $r$ at the BS, which denotes its response on the generated multiple analog beams of RF chain $r$.
Through the uplink pilot transmission, the effective channels of all the users on all the RF chains can be estimated at the BS.
In the following, we denote the effective channel vector of user $k$ as ${{{\mathbf{\tilde h}}}_k} = {\left[ {\begin{array}{*{20}{c}}
{{{\tilde h}_{k,1}}}& \cdots &{{{\tilde h}_{k,{N_{{\mathrm{RF}}}}}}}
\end{array}} \right]} \in \mathbb{C}^{{ 1\times N_{\mathrm{RF}}}}$ and denote the effective channel matrix between the BS and the $K$ users as ${\mathbf{\tilde H}} = {\left[ {\begin{array}{*{20}{c}}
{{\mathbf{h}}_1^{\mathrm{T}}}& \cdots &{{\mathbf{h}}_K^{\mathrm{T}}}
\end{array}} \right]^{\mathrm{T}}} \in \mathbb{C}^{{ K \times N_{\mathrm{RF}}}}$.
\vspace{-5mm}
\subsection{Digital Precoder and Power Allocation Design}
\vspace{-1mm}
Given the estimated effective channel matrix ${{\mathbf{\tilde H}}}$, the digital precoder and the power allocation can be designed accordingly.
Note that, all the users in a NOMA group share the same digital precoder and there are totally $N_{\mathrm{RF}}$ groups of users to be served in the proposed multi-beam NOMA scheme.
Assuming that the adopted digital precoder is denoted as ${\mathbf{G}} = \left[ {\begin{array}{*{20}{c}}
{{{\mathbf{g}}_1}}& \cdots &{{{\mathbf{g}}_{{N_{{\mathrm{RF}}}}}}}
\end{array}} \right] \in \mathbb{C}^{{ N_{\mathrm{RF}} \times N_{\mathrm{RF}}}}$, where ${{\mathbf{g}}_r}$ denotes the digital precoder for the NOMA group associated with RF chain $r$.
%
%
In addition, we assume that the designed power allocation for downlink transmission of user $k$ associated with RF chain $r$ at the BS is denoted as $p_{k,r}$, with $\sum\nolimits_{k = 1}^K {{u_{k,r}}{p_{k,r}}}  \le p_{\mathrm{max}}$, $\forall r$.
Then, the received signal at user $k$ is given by
\vspace{-2mm}
\begin{equation}\label{DLRx1}
{y_k} = {{{\mathbf{\tilde h}}}_k}{{\mathbf{G}}\mathbf{t}} + {z_k} = {{{\mathbf{\tilde h}}}_k}\sum\nolimits_{r = 1}^{{N_{{\mathrm{RF}}}}} {{{\mathbf{g}}_r}{t_r}}  + {z_k},\vspace{-2mm}
\end{equation}
where ${\mathbf{t}} = {\left[
{{t_1}}, \ldots ,{{t_{{N_{{\mathrm{RF}}}}}}}
\right]^{\mathrm{T}}} \in \mathbb{C}^{{ N_{\mathrm{RF}} \times 1}}$ denotes the superimposed signals of all NOMA groups with ${t_r} = \sum\nolimits_{k = 1}^K {{u_{k,r}}} \sqrt{p_{k,r}}{s_k}$ denoting the superimposed signal of the NOMA group associated with RF chain $r$.
Variable ${s_k} \in \mathbb{C}$ denotes the modulated symbol for user $k$ and $z_k \sim {\cal CN}(0,\sigma^2)$ denotes the additive white Gaussian noise (AWGN) at user $k$, where ${\cal CN}(0,\sigma^2)$ denotes the circularly symmetric complex Gaussian distribution with zero-mean and variance of $\sigma^2$.
For instance, in Figure \ref{MultiBeamNOMA:a}, if user $k$ and user $j$ are allocated to RF chain $r$ and user $i$ is allocated to RF chain $r'$, we have ${t_r} = \sqrt{p_{k,r}}{s_k} + \sqrt{p_{j,r}}{s_j}$ and ${t_{r'}} = \sqrt{p_{i,r'}}{s_i}$.

Note that, the conventional precoder design method, such as the zero-forcing (ZF)\cite{zhao2017multiuser}, can be utilized in our proposed multi-beam NOMA scheme to design ${\mathbf{G}}$.
However, this paper focuses on the pure analog beamforming for hybrid mmWave systems\cite{GaoSubarray} to study the performance gain achieved by generating multiple analog beams\footnote{Note that the use of digital precoder will further improve the performance of the proposed multi-beam NOMA scheme due to the mitigation of interferences among different NOMA groups. However, due to the page limitation, we will consider it in our future work.}.
For the pure analog beamforming in hybrid mmWave systems, the digital precoder is given by ${\mathbf{G}} = {{\mathbf{I}}_{{N_{{\mathrm{RF}}}}}}$ and each RF chain serves its associated NOMA group correspondingly.
Now, when user $k$ is allocated with RF chain $r$ at the BS, i.e., ${u_{k,r}} = 1$, the received signal at user $k$ in \eqref{DLRx1} can be rewritten as
\vspace{-2mm}
\begin{align}\label{DLRx2}
{y_k}&= \underbrace{{{\tilde h}_{k,r}}\sqrt {{p_{k,r}}} {s_k}}_{\text{Desired signal}} + \underbrace{{{\tilde h}_{k,r}}\sum\nolimits_{k' \ne k}^K {{u_{k',r}}\sqrt {{p_{k',r}}} {s_{k'}}}}_{\text{Intra-group interference}} \notag\\[-1mm]
& + \underbrace{\sum\nolimits_{r' \ne r}^{{N_{{\mathrm{RF}}}}} {{{\tilde h}_{k,r'}}\sum\nolimits_{k' = 1}^K {{u_{k',r'}}} \sqrt {{p_{k',r'}}} {s_{k'}}}}_{\text{Inter-group interference}}  + {z_k} ,
\end{align}
\par
\vspace{-2mm}
\noindent
where the first term denotes the desired signal of user $k$, the second term denotes the intra-group interference caused by the other users within the NOMA group associated with RF chain $r$, and the third term denotes the inter-group interference originated from all the other RF chains.
On one hand, owing to the analog beamforming, the effective channel gain between user $k$ and all the other RF chains, $\forall {r'} \neq r$, $\left|{{\tilde h}_{k,r'}}\right|$, is generally very small.
Therefore, the inter-group interference is limited.
On the other hand, the intra-group interference can be handled by the SIC decoding as detailed in the following.
\vspace{-1mm}
\subsection{SIC Decoding at Users}
\vspace{-1mm}
At users side, as the traditional downlink NOMA schemes\cite{Wei2016NOMA,WeiTCOM2017}, SIC decoding is performed at the strong users within one NOMA group, while the weak users directly decode the messages by treating the signals of strong users as noise.
In this paper, we define the strong or weak user by the LOS path gain.
Without loss of generality, we assume that the users are indexed in the descending order of LOS path gains, i.e., ${\left| {{\alpha _{1,0}}} \right|^2} \ge {\left| {{\alpha _{2,0}}} \right|^2} \ge , \ldots , \ge {\left| {{\alpha _{K,0}}} \right|^2}$.
Then, user $1$ has the strongest channel on average, while user $K$ possesses the weakest channel on average.

According to the downlink NOMA protocol\cite{WeiTCOM2017}, user $k$ first decodes the message of user $K$ and subtract the signal of user $K$ from its received signal, and then perform SIC sequentially in a similar way as for user $K-1, \ldots, k+1$.
In other words, in \eqref{DLRx2}, the intra-group interference from the other users with $\forall k' > k$ within the same NOMA group can be eliminated via SIC decoding at user $k$.
Therefore, the individual data rate of user $k$ when associated with RF chain $r$ is given by
\vspace{-2mm}
\begin{equation}\label{DLIndividualRate1}
{R_{k,r}} = {\log _2}\left( {1 + \frac{{{u_{k,r}}{p_{k,r}}{{\left| {{{\tilde h}_{k,r}}} \right|}^2}}}{{I_{k,r}^{{{\mathrm{ inter}}}} + I_{k,r}^{{{\mathrm{intra}}}} + \sigma^2}}} \right),\vspace{-2mm}
\end{equation}
with
\vspace{-2mm}
\begin{align}
I_{k,r}^{{{\mathrm{inter}}}} &= \sum\nolimits_{r' \ne r}^{{N_{{\mathrm{RF}}}}} {{{\left| {{{\tilde h}_{k,r'}}} \right|}^2}\sum\nolimits_{k' = 1}^K {{u_{k',r'}}{p_{k',r'}}} } \;\mathrm{and}\;\notag\\[-0.5mm]
I_{k,r}^{{{\mathrm{intra}}}}  &= \sum\nolimits_{k' = 1}^{k-1} {{u_{k',r}}{p_{k',r}}{{\left| {{{\tilde h}_{k,r}}} \right|}^2}}
\end{align}
\par
\vspace{-2mm}
\noindent
denoting the inter-group interference power and intra-group interference power, respectively.
Note that with the formulation in \eqref{DLIndividualRate1}, we have ${R_{k,r}} = 0$ if ${u_{k,r}} = 0$.
During SIC decoding of the message of user $k'$ at user $k$, $\forall k' > k$, the achievable data rate is given by
\vspace{-2mm}
\begin{equation}\label{IndividualRate2}
R_{k,k',r} = {\log _2}\left( {1 + \frac{{{u_{k',r}}{p_{k',r}}{{\left| {{{\tilde h}_{k,r}}} \right|}^2}}}{{I_{k,r}^{{\mathrm{inter}}} + I_{k,k',r}^{{\mathrm{intra}}} + \sigma^2}}} \right),\vspace{-2mm}
\end{equation}
where $I_{k,k',r}^{{\mathrm{intra}}} = \sum\nolimits_{k'' = 1}^{k' - 1} {{u_{k'',r}}{p_{k'',r}}{{\left| {{{\tilde h}_{k,r}}} \right|}^2}}$ denotes the intra-group interference power when decoding the message of user $k'$ at user $k$.
To guarantee the success of the SIC decoding, we need to maintain the rate condition as follows \cite{Sun2016Fullduplex}
\vspace{-2mm}
\begin{equation}\label{DLSICDecoding}
R_{k,k',r} \ge {R_{k',r}}, \forall k' > k.\vspace{-2mm}
\end{equation}
Note that, when user $k'$ is not allocated with RF chain $r$, we have $R_{k,k',r} = {R_{k',r}} =0$ and this condition \eqref{DLSICDecoding} will always be satisfied.
Now, the individual rate of user $k$ is defined as $R_{k} = {\sum\nolimits_{r = 1}^{{N_{\mathrm{RF}}}} {{R_{k,r}}}}$, $\forall k$, and the system sum-rate is given by
\vspace{-2mm}
\begin{align}\label{SumRate}
{R_{{\mathrm{sum}}}} = \sum\limits_{k = 1}^K {\sum\limits_{r = 1}^{{N_{RF}}} {{{\log }_2}\left( {1 + \frac{{{u_{k,r}}{p_{k,r}}{{\left| {{{\tilde h}_{k,r}}} \right|}^2}}}{{I_{k,r}^{{{\mathrm{inter}}}} + I_{k,r}^{{{\mathrm{intra}}}} + \sigma^2}}} \right)} }.
\end{align}
\par
\vspace{-2mm}
\noindent

Different from the single-beam NOMA schemes for
mmWave systems\cite{Ding2017RandomBeamforming,Cui2017Optimal, WangBeamSpace2017} and conventional NOMA schemes for microwave systems\cite{WeiTCOM2017}, the system sum-rate in \eqref{SumRate} depends on the effective channel gains, which varies according to the adopted user grouping and antenna allocation strategy.

\section{A Single-RF Chain System}
To reveal more insights of the proposed multi-beam NOMA framework, in this section, we consider a simple single-RF chain system where a single RF chain is equipped at each transceiver.
In particular, we study the asymptotic performance analysis of the system sum-rate for the proposed scheme, in the large number of antennas and high signal-to-noise ratio (SNR) regimes.
Furthermore, we compare the asymptotic system sum-rate between the proposed multi-beam NOMA scheme and the conventional TDMA scheme, and derive the antenna allocation conditions that the proposed scheme outperforms TDMA in terms of system sum-rate.

\vspace{-1mm}
\subsection{Asymptotic Performance Analysis}
\vspace{-1mm}
We consider a single-RF chain BS serving $K$ users allocated with this RF chain.
To facilitate the following presentation, the subscript of RF chain $r$ is omitted in this section.
In this section, we focus on the asymptotic system performance when the number of antennas equipped at the BS and each user as well as the allocated number of antennas for user $k$ are sufficiently large, i.e., $M_{\mathrm{BS}}, M_{\mathrm{UE}}, M_k \to \infty$.
As a result, the effective channel gain of user $k$, ${\tilde h_{k}}$ in \eqref{EffectiveChannel2} can be approximated by
\vspace{-3mm}
\begin{equation}\label{EffectiveChannel3}
{\tilde h_{k}} \mathop  \approx \limits_{{M_{{\mathrm{BS}}}},{M_{{\mathrm{UE}}}},{M_k} \to \infty }  {\alpha _{k,0}}{\sqrt {\frac{{{M_{{\mathrm{UE}}}}}}{{{M_{\mathrm{BS}}}}}} } M_k, \forall k.\vspace{-2mm}
\end{equation}
Note that the approximation in \eqref{EffectiveChannel3} is obtained via \eqref{EffectiveChannel2} by applying $\mathop {\lim }\limits_{{M_{k'}} \to \infty } \frac{{\sin \left( {{M_{k'}}{\varphi _{k,k',l}}} \right)}}{{\sin \left( {{\varphi _{k',k}}} \right)}} \to 0$ with ${\varphi _{k,k',l}} > 0$ and $\mathop {\lim }\limits_{{M_{\mathrm{UE}}} \to \infty } \frac{{\sin \left( {{M_{{\mathrm{UE}}}}{\kappa _{k,l}}} \right)}}{{\sin \left( {{\kappa _{k,l}}} \right)}} \to 0$ with ${\kappa _{k,l}} > 0$.
In other words, only the LOS path contributes to the effective channel \eqref{EffectiveChannel3} due to the narrow analog beamwidth in the large antenna regime.

With a single RF chain and a single NOMA group, there is no inter-group interference $I_{k}^{{{\mathrm{inter}}}}$ in \eqref{DLIndividualRate1} and \eqref{IndividualRate2}.
Therefore, the rate condition of SIC decoding in \eqref{DLSICDecoding} is equivalent to
\vspace{-2mm}
\begin{equation}\label{DLSICDecoding2}
\left| {{\alpha _{k,0}}} \right|{M_k} \ge \left| {{\alpha _{k',0}}} \right|{M_{k'}}, \forall k' \ge k,\vspace{-2mm}
\end{equation}
which means that the antenna allocation should maintain the strong user having a larger effective channel gain compared to weak users.
Then, the asymptotic performance of the individual rates and the system sum-rate is summarized in the following theorem.

\begin{Thm}\label{NOMASumRateApproximate}
With successful SIC decoding, the asymptotic individual data rates for the proposed multi-beam NOMA scheme can be approximated by
\vspace{-2mm}
\begin{align}
\hspace{-3mm}R_1^{{\mathrm{NOMA}}} & \mathop  \approx \limits_{\rho  \to \infty }  {\log _2}\hspace{-1mm}\left( {\frac{{\rho {p_1}}{{\left| {{\alpha _{1,0}}} \right|}^2}{{M_{{\mathrm{UE}}}}M_1^2}}{{{M_{{\mathrm{BS}}}}}}} \right)\; \mathrm{and}\label{NOMAIndividualRateAppro21}\\[-0.5mm]
\hspace{-3mm}R_k^{{\mathrm{NOMA}}} &\mathop  \approx \limits_{\rho  \to \infty } {\log _2}\hspace{-1mm}\left( {1 \hspace{-1mm}+\hspace{-1mm} \frac{{{p_k}}}{{\sum\nolimits_{k' = 1}^{k - 1} {{p_{k'}}} }}} \right), \forall k = \{2,\ldots,K\},\label{NOMAIndividualRateAppro22}
\end{align}
\par
\vspace{-2mm}
\noindent
with the asymptotic system sum-rate approximated by
\vspace{-2mm}
\begin{equation}\label{NOMASumRateAppro1}
R_{\mathrm{sum}}^{{\mathrm{NOMA}}} \mathop  \approx \limits_{\rho  \to \infty } {\log _2}\left( { {\frac{\rho {p_{{\mathrm{max}}}}{{\left| {{\alpha _{1,0}}} \right|}^2}{{M_{{\mathrm{UE}}}}M_1^2}}{{{M_{{\mathrm{BS}}}}}}}} \right),\vspace{-3mm}
\end{equation}
where $\rho  = \frac{1}{{{\sigma ^2}}} \to \infty$.
\end{Thm}
\begin{proof}
Assuming the antenna allocation satisfying the condition in \eqref{DLSICDecoding2}, according to \eqref{DLIndividualRate1}, the individual data rate of each NOMA user can be rewritten as
\vspace{-2mm}
\begin{equation}\label{NOMAIndividualRateAppro1}
R_k^{{\mathrm{NOMA}}} = {\log _2}\left( {1 + \frac{{\rho {p_k}{{\left| {{{\tilde h}_k}} \right|}^2}}}{{\rho \sum\nolimits_{k' = 1}^{k - 1} {{p_{k'}}{{\left| {{{\tilde h}_{k}}} \right|}^2}}  + 1}}} \right), \forall k.\vspace{-2mm}
\end{equation}
Substituting the effective channel expression from \eqref{EffectiveChannel3} into \eqref{NOMAIndividualRateAppro1}, the analytical results in \eqref{NOMAIndividualRateAppro21}, \eqref{NOMAIndividualRateAppro22}, and \eqref{NOMASumRateAppro1} can be easily obtained with the high SNR approximation.
\end{proof}

From \eqref{EffectiveChannel3} and \eqref{NOMASumRateAppro1}, it can be observed that the asymptotic system sum-rate is only determined by the effective channel gain of user $1$, ${{\left| {{{\tilde h}_1}} \right|}^2}$, and hence is determined by the number of antennas allocated for the strongest user.
%
%
From \eqref{NOMAIndividualRateAppro22}, we can observe that the individual rates of weak users, $k = 2,\ldots,K$, are determined by the power allocation ratio, $\frac{{{p_k}}}{{\sum\nolimits_{k' = 1}^{k - 1} {{p_{k'}}} }}$.
In contrast, from \eqref{NOMAIndividualRateAppro21}, the individual rate of the strongest user $R_1$ is determined by not only the power allocation $p_1$, but also the allocated number of antennas $M_1$.
%
%
Therefore, the strongest user prefers more antennas than the weak users, while weak user $k$ prefers more power compared to the strong users, $k' < k$.
It is due to the fact that the weak users suffer from inter-user interference (IUI) as in \eqref{NOMAIndividualRateAppro1}, while the strongest user is IUI free owing to the SIC decoding.
Hence, allocating more antennas to weak users not only increases the strength of their own received signal, but also increases the IUI, which only leads to a marginal gain in data rate.
In contrast, allocating more antennas to the strongest user will directly increase its received signal power and individual rate.
%
%
\vspace{-2mm}
\subsection{Comparison between NOMA and TDMA}
\vspace{-1mm}
Considering TDMA as the baseline OMA scheme, the individual rate and system sum-rate of TDMA are given by
\vspace{-2mm}
\begin{align}
R_k^{{\mathrm{TDMA}}} &= {\beta_k}{\log _2}\left( {1 + \rho p_{\mathrm{max}}{{\left| {{h_k}} \right|}^2}} \right), \forall k, \;\mathrm{and} \notag\\[-1mm]
R_{\mathrm{sum}}^{{\mathrm{TDMA}}} &= \sum\nolimits_{k = 1}^K {{\beta_k}{{\log }_2}\left( {1 + \rho p_{\mathrm{max}}{{\left| {{h_k}} \right|}^2}} \right)},
\end{align}
\par
\vspace{-2mm}
\noindent
respectively, where ${\beta_k}$ denotes the allocated time for user $k$ and ${{\left| {{h_k}} \right|}^2}$ denotes the effective channel of user $k$ without beam splitting.
With $M_{\mathrm{BS}} \to \infty$, ${{\left| {{h_k}} \right|}^2}$ can be approximated by ${{\left| {{h_k}} \right|}^2} \approx {{\left| {\alpha _{k,0}} \right|}^2}{ {{{{M_{{\mathrm{UE}}}}}}{{{M_{\mathrm{BS}}}}}} }$, $\forall k$.
Similar to \eqref{NOMASumRateAppro1}, the high SNR approximation of the system sum-rate of TDMA is given by
\vspace{-1mm}
\begin{equation}\label{OMASumRateAppro1}
\hspace{-2mm}R_{\mathrm{sum}}^{{\mathrm{TDMA}}} \mathop  \approx \limits_{\rho  \to \infty } \sum\nolimits_{k = 1}^K {{\beta_k}{{\log }_2}\left( { \rho p_{\mathrm{max}}{{\left| {\alpha _{k,0}} \right|}^2}{ {{{{M_{{\mathrm{UE}}}}}}{{{M_{\mathrm{BS}}}}}} }} \right)}.\vspace{-2mm}
\end{equation}

To analyze the effect of beam splitting on the performance of the proposed multi-beam NOMA scheme, we restrict ourselves to the case of equal power allocation for the proposed multi-beam NOMA scheme as well as equal
time allocation for the TDMA scheme, i.e., ${p_k} = \frac{1}{K}p_{\mathrm{max}}$ and ${\beta_k} = \frac{1}{K}$.
Now, we introduce the following corollaries which can be obtained from Theorem \ref{NOMASumRateApproximate}.

\begin{Cor}
The performance gain of the proposed multi-beam NOMA scheme over the TDMA scheme is given by
\vspace{-5mm}
\begin{align}\label{PerformanceGain}
R_{\mathrm{sum}}^{{\mathrm{NOMA}}} - R_{\mathrm{sum}}^{{\mathrm{TDMA}}}
= 2{\log _2}\left( {\frac{{{M_1}}}{{{M_{{\mathrm{BS}}}}}}\frac{{\left| {{\alpha _{1,0}}} \right|}}{\overline{{\left| {{\alpha _{0}}} \right|}}}} \right),
\end{align}
\par
\vspace{-2.5mm}
\noindent
where ${\overline{{\left| {{\alpha _{0}}} \right|}}} = {{{{\left( {\mathop \Pi \limits_{k = 1}^K {{\left| {{\alpha _{k,0}}} \right|}}} \right)}^{\frac{1}{K}}}}}$ denotes the geometric mean of all the users' LOS path gain.
\end{Cor}

\begin{Cor}
The proposed multi-beam NOMA scheme provides a higher system sum-rate over TDMA when \eqref{DLSICDecoding2} and the following condition are satisfied:
\vspace{-2mm}
\begin{equation}\label{BreakEvenCondition}
\left| {{\alpha _{1,0}}} \right|{M_1} > {M_{{\mathrm{BS}}}}\overline{\left| {{\alpha _{0}}} \right|}.\vspace{-2mm}
\end{equation}
\end{Cor}

Due to page limitation, the proofs for the corollaries are omitted.
Furthermore, it can be easily verified that the conditions \eqref{DLSICDecoding2} and \eqref{BreakEvenCondition} are sufficient and necessary conditions to guarantee the proposed scheme outperforming the TDMA scheme in terms of system sum-rate.
\vspace{-1mm}
\section{Numerical Results}
\vspace{-1mm}
In this section, we evaluate the performance of our proposed multi-beam NOMA scheme in a single-RF chain system through simulations.
We consider $L=30$ paths for the channel model in \eqref{ChannelModel1}, where the LOS and NLOS path gains are generated according to corresponding models in \cite{AkdenizChannelMmWave}, respectively.
The number of antennas equipped at the BS and users are ${M_{{\mathrm{BS}}}} = 128$ and ${M_{{\mathrm{UE}}}} = 10$, respectively, the maximum transmit power at the BS is $p_{\mathrm{max}} = 46$ dBm, and the noise variance is $\sigma^2 = -88$ dBm.
All the $K$ users are randomly deployed in a cell with the cell size of $500$ m.
\vspace{-1mm}
\subsection{Multi-Beam NOMA versus TDMA}
\vspace{-1mm}
This simulation compares the performance between the proposed multi-beam NOMA scheme and the conventional TDMA scheme.
We consider a two-user case, i.e, $K=2$, and assume that user $1$ is the strong user while user $2$ is the weak user, i.e., ${\left| {{\alpha _{1,0}}} \right|} \ge {\left| {{\alpha _{2,0}}} \right|}$.
Two simulation cases with $\frac{{\left| {{\alpha _{1,0}}} \right|}}{{\left| {{\alpha _{2,0}}} \right|}} = 5$ and $\frac{{\left| {{\alpha _{1,0}}} \right|}}{{\left| {{\alpha _{2,0}}} \right|}} = 10$ are considered.
Figure 4 illustrates the average system sum-rate versus the number of antennas allocated to user 1, $M_1$, for the proposed scheme and the TDMA scheme.
Note that the search range of $M_1$ is kept to satisfy the conditions in \eqref{DLSICDecoding2} and \eqref{BreakEvenCondition}.
The threshold predicted with \eqref{BreakEvenCondition} is illustrated with vertical red lines.
As mentioned before, we consider equal power allocation for the proposed multi-beam NOMA scheme and equal time allocation for the TDMA scheme.
It can be observed that the simulation results match well with the developed asymptotic analysis for both schemes.
Besides, as predicted in \eqref{NOMASumRateAppro1}, the system sum-rate of the proposed multi-beam NOMA scheme increases monotonically with $M_1$.
Furthermore, the antenna allocation condition in \eqref{BreakEvenCondition} can accurately predict the threshold for the number of antennas allocated for user 1, such that the proposed multi-beam NOMA scheme provides a higher system sum-rate than that of the TDMA scheme.
Comparing the two simulation cases of $\frac{{\left| {{\alpha _{1,0}}} \right|}}{{\left| {{\alpha _{2,0}}} \right|}} = 5$ and $\frac{{\left| {{\alpha _{1,0}}} \right|}}{{\left| {{\alpha _{2,0}}} \right|}} = 10$, the performance gain of the proposed multi-beam NOMA scheme over the TDMA scheme in the second case is higher than that of the first case.
This is due to the fact that NOMA can efficiently exploit the disparity of the channel gains for improving the system performance.

\begin{figure}[t]
\vspace{-3mm}
\centering
\includegraphics[width=3in]{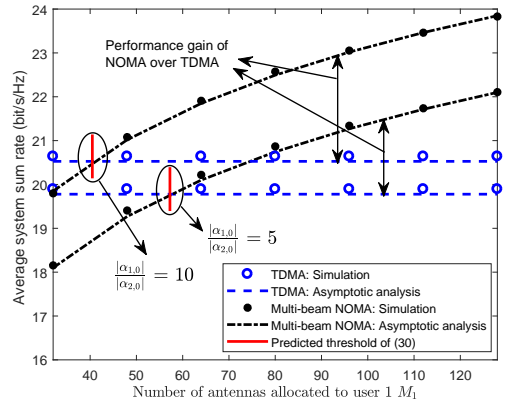}\vspace{-3mm}
\caption{Average system sum-rate  versus $M_1$ for the proposed multi-beam NOMA scheme and the TDMA scheme in a single-RF chain system.}\vspace{-6mm}
\label{SingleRFChainSystemsThreshold}
\end{figure}

\vspace{-1mm}
\subsection{Multi-Beam NOMA versus Single-Beam NOMA}
\vspace{-1mm}
This section investigates the performance gain of the proposed multi-beam NOMA scheme over a baseline scheme utilizing single-beam NOMA \cite{Ding2017RandomBeamforming,Cui2017Optimal,WangBeamSpace2017} for a scenario with $K = 5$ users.
In particular, for the baseline scheme, when multiple users' AODs are within the same analog beam, we adopt the single-beam NOMA\cite{Ding2017RandomBeamforming,Cui2017Optimal,WangBeamSpace2017} to serve them simultaneously, otherwise the BS serves them using TDMA with equal time allocation.
For the proposed multi-beam NOMA scheme, we allocate $M_1 = 100$ antennas to the strongest user, and the remaining $28$ antennas are equally allocated to the other four users with $M_2 = M_3 = M_4 = M_5 = 7$.
Note that, this antenna allocation strategy can satisfy the antenna allocation conditions \eqref{DLSICDecoding2} and \eqref{BreakEvenCondition} with a high probability.
Figure 5 shows the average sum-rate versus the maximum transmit power at the BS $p_{\mathrm{max}}$ for the proposed multi-beam NOMA scheme, the baseline scheme, and the TDMA scheme.
Equal power allocation is adopted for all the three schemes and equal time allocation is adopted for the TDMA scheme.
It can be observed that the baseline scheme is only slightly better than the TDMA scheme.
In fact, due to the large number of antennas array equipped at the BS, the probability of multiple users that fall in the  main lobe of an analog beam is very small due to the narrow analog beamwidth.
Therefore, the contribution of the single-beam NOMA to the system performance enhancement of the baseline scheme compared to the TDMA scheme is limited.
In contrast, the proposed multi-beam scheme enables NOMA transmission to multiple users with an arbitrary AODs distribution, which can achieve a considerable performance gain over the baseline scheme.
In addition, a constant performance gain can be achieved by the proposed multi-beam NOMA scheme over the TDMA scheme when $p_{\mathrm{max}}$ increases.
In fact, owing to the massive number of antennas, the considered hybrid mmWave system operate at high SNR regime.
Hence, the performance gain of multi-beam NOMA over TDMA can be accurately predicted by \eqref{PerformanceGain}, which no longer depends on $p_{\mathrm{max}}$.

\begin{figure}[t]
\vspace{-3mm}
\centering
\includegraphics[width=3in]{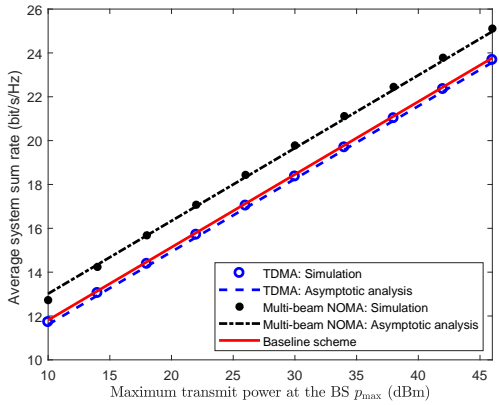}\vspace{-4mm}
\caption{Average system sum-rate versus maximum transmit power at the BS $p_{\mathrm{max}}$ for the proposed multi-beam NOMA scheme, the baseline scheme, and the TDMA scheme in a single-RF chain system.}\vspace{-6mm}
\label{SingleRFChainSystemsThreshold}
\end{figure}

\vspace{-1mm}
\section{Conclusion}
\vspace{-0.5mm}
In this paper, a novel multi-beam NOMA framework was proposed for hybrid mmWave systems to accommodate more users with a limited number of RF chains.
Compared to the single-beam NOMA scheme, the proposed multi-beam NOMA scheme enjoys not only the higher spectral efficiency but also the flexibility for serving multiple users.
The multiple analog beams are generated using the proposed beam splitting technique, which dynamically divides the whole antenna array into multiple subarrays and each subarray generates an analog beam.
For a single-RF chain system, we analyzed the asymptotic system sum-rate in the large number of antennas and high SNR regimes.
Furthermore, the sufficient and necessary conditions of antenna allocation,  which guarantee the proposed scheme outperforming the conventional TDMA scheme, were obtained.
Simulation results confirmed the accuracy of our developed analysis and demonstrated that our proposed multi-beam NOMA scheme provides a higher spectral efficiency than the single-beam NOMA scheme as well as the TDMA scheme.


\begin{thebibliography}{10}
\providecommand{\url}[1]{#1}
\csname url@samestyle\endcsname
\providecommand{\newblock}{\relax}
\providecommand{\bibinfo}[2]{#2}
\providecommand{\BIBentrySTDinterwordspacing}{\spaceskip=0pt\relax}
\providecommand{\BIBentryALTinterwordstretchfactor}{4}
\providecommand{\BIBentryALTinterwordspacing}{\spaceskip=\fontdimen2\font plus
\BIBentryALTinterwordstretchfactor\fontdimen3\font minus
  \fontdimen4\font\relax}
\providecommand{\BIBforeignlanguage}[2]{{%
\expandafter\ifx\csname l@#1\endcsname\relax
\typeout{** WARNING: IEEEtran.bst: No hyphenation pattern has been}%
\typeout{** loaded for the language `#1'. Using the pattern for}%
\typeout{** the default language instead.}%
\else
\language=\csname l@#1\endcsname
\fi
#2}}
\providecommand{\BIBdecl}{\relax}
\BIBdecl

\bibitem{Rappaport2013}
T.~Rappaport, S.~Sun, R.~Mayzus, H.~Zhao, Y.~Azar, K.~Wang, G.~Wong, J.~Schulz,
  M.~Samimi, and F.~Gutierrez, ``Millimeter wave mobile communications for {5G}
  cellular: {It} will work!'' \emph{IEEE Access}, vol.~1, pp. 335--349, May
  2013.

\bibitem{XiaoMing2017}
M.~Xiao, S.~Mumtaz, Y.~Huang, L.~Dai, Y.~Li, M.~Matthaiou, G.~K. Karagiannidis,
  E.~Bj{\"{o}}rnson, K.~Yang, C.~L. I, and A.~Ghosh, ``Millimeter wave
  communications for future mobile networks,'' \emph{IEEE J. Select. Areas
  Commun.}, vol.~35, no.~9, pp. 1909--1935, Sep. 2017.

\bibitem{wong2017key}
V.~W. Wong, R.~Schober, D.~W.~K. Ng, and L.-C. Wang, \emph{Key Technologies for
  {5G} Wireless Systems}.\hskip 1em plus 0.5em minus 0.4em\relax Cambridge
  University Press, 2017.

\bibitem{zhao2017multiuser}
L.~Zhao, D.~W.~K. Ng, and J.~Yuan, ``Multi-user precoding and channel
  estimation for hybrid millimeter wave systems,'' \emph{IEEE J. Select. Areas
  Commun.}, vol.~35, no.~7, Jul. 2017.

\bibitem{GaoSubarray}
X.~Gao, L.~Dai, S.~Han, C.~L. I, and R.~W. Heath, ``Energy-efficient hybrid
  analog and digital precoding for {MmWave} {MIMO} systems with large antenna
  arrays,'' \emph{IEEE J. Select. Areas Commun.}, vol.~34, no.~4, pp.
  998--1009, Apr. 2016.

\bibitem{lin2016energy}
C.~Lin and G.~Y. Li, ``Energy-efficient design of indoor mmwave and {sub-THz}
  systems with antenna arrays,'' \emph{IEEE Trans. Wireless Commun.}, vol.~15,
  no.~7, pp. 4660--4672, Mar. 2016.

\bibitem{van2002optimum}
H.~L. Van~Trees, \emph{Optimum array processing: Part {IV} of detection,
  estimation and modulation theory}.\hskip 1em plus 0.5em minus 0.4em\relax
  Wiley Online Library, 2002, vol.~1.

\bibitem{Ding2015b}
Z.~Ding, Y.~Liu, J.~Choi, Q.~Sun, M.~Elkashlan, C.~L. I, and H.~V. Poor,
  ``Application of non-orthogonal multiple access in {LTE} and {5G} networks,''
  \emph{IEEE Commun. Mag.}, vol.~55, no.~2, pp. 185--191, Feb. 2017.

\bibitem{Andrews2014}
J.~Andrews, S.~Buzzi, W.~Choi, S.~Hanly, A.~Lozano, A.~Soong, and J.~Zhang,
  ``What will {5G} be?'' \emph{IEEE J. Select. Areas Commun.}, vol.~32, no.~6,
  pp. 1065--1082, Jun. 2014.

\bibitem{Dai2015}
L.~Dai, B.~Wang, Y.~Yuan, S.~Han, I.~Chih-Lin, and Z.~Wang, ``Non-orthogonal
  multiple access for {5G}: solutions, challenges, opportunities, and future
  research trends,'' \emph{IEEE Commun. Mag.}, vol.~53, no.~9, pp. 74--81, Sep.
  2015.

\bibitem{WeiSurvey2016}
Z.~Wei, Y.~Jinhong, D.~W.~K. Ng, M.~Elkashlan, and Z.~Ding, ``A survey of
  downlink non-orthogonal multiple access for {5G} wireless communication
  networks,'' \emph{ZTE Commun.}, vol.~14, no.~4, pp. 17--25, Oct. 2016.

\bibitem{Kwan_AF_2010}
D.~W.~K. Ng and R.~Schober, ``Cross-layer scheduling for {OFDMA}
  amplify-and-forward relay networks,'' \emph{IEEE Trans. Veh. Technol.},
  vol.~59, no.~3, pp. 1443--1458, Mar. 2010.

\bibitem{DerrickEEOFDMA}
D.~W.~K. Ng, E.~S. Lo, and R.~Schober, ``Energy-efficient resource allocation
  in {OFDMA} systems with large numbers of base station antennas,'' \emph{IEEE
  Trans. Wireless Commun.}, vol.~11, no.~9, pp. 3292--3304, Sep. 2012.

\bibitem{Ding2014}
Z.~Ding, Z.~Yang, P.~Fan, and H.~Poor, ``On the performance of non-orthogonal
  multiple access in {5G} systems with randomly deployed users,'' \emph{IEEE
  Signal Process. Lett.}, vol.~21, no.~12, pp. 1501--1505, Dec. 2014.

\bibitem{Ding2017RandomBeamforming}
Z.~Ding, P.~Fan, and H.~V. Poor, ``Random beamforming in millimeter-wave {NOMA}
  networks,'' \emph{IEEE Access}, vol.~5, pp. 7667--7681, Feb. 2017.

\bibitem{Cui2017Optimal}
J.~Cui, Y.~Liu, Z.~Ding, P.~Fan, and A.~Nallanathan, ``Optimal user scheduling
  and power allocation for millimeter wave {NOMA} systems,'' \emph{arXiv
  preprint arXiv:1705.03064}, 2017.

\bibitem{WangBeamSpace2017}
B.~Wang, L.~Dai, Z.~Wang, N.~Ge, and S.~Zhou, ``Spectrum and energy efficient
  beamspace {MIMO-NOMA} for millimeter-wave communications using lens antenna
  array,'' \emph{IEEE J. Select. Areas Commun.}, vol.~PP, no.~99, pp. 1--1,
  Jul. 2017.

\bibitem{Xiao2017MultiBeam}
Z.~Xiao, L.~Dai, Z.~Ding, J.~Choi, and P.~Xia, ``Millimeter-wave communication
  with non-orthogonal multiple access for {5G},'' \emph{arXiv preprint
  arXiv:1709.07980}, 2017.

\bibitem{BeamTracking}
V.~Va, H.~Vikalo, and R.~W. Heath, ``Beam tracking for mobile millimeter wave
  communication systems,'' in \emph{Proc. IEEE Global Conf. on Signal and Inf.
  Process.}, Dec. 2016, pp. 743--747.

\bibitem{Hanif2016}
M.~F. Hanif, Z.~Ding, T.~Ratnarajah, and G.~K. Karagiannidis, ``A
  minorization-maximization method for optimizing sum rate in the downlink of
  non-orthogonal multiple access systems,'' \emph{IEEE Trans. Signal Process.},
  vol.~64, no.~1, pp. 76--88, Jan. 2016.

\bibitem{Wei2016NOMA}
Z.~Wei, D.~W.~K. Ng, and J.~Yuan, ``Power-efficient resource allocation for
  {MC-NOMA} with statistical channel state information,'' in \emph{Proc. IEEE
  Global Commun. Conf.}, Dec. 2016, pp. 1--7.

\bibitem{WeiTCOM2017}
Z.~Wei, D.~W.~K. Ng, J.~Yuan, and H.~M. Wang, ``Optimal resource allocation for
  power-efficient {MC-NOMA} with imperfect channel state information,''
  \emph{IEEE Trans. Commun.}, vol.~PP, no.~99, pp. 1--1, May 2017.

\bibitem{Sun2016Fullduplex}
Y.~Sun, D.~W.~K. Ng, Z.~Ding, and R.~Schober, ``Optimal joint power and
  subcarrier allocation for full-duplex multicarrier non-orthogonal multiple
  access systems,'' \emph{IEEE Trans. Commun.}, vol.~65, no.~3, pp. 1077--1091,
  Mar. 2017.

\bibitem{AkdenizChannelMmWave}
M.~R. Akdeniz, Y.~Liu, M.~K. Samimi, S.~Sun, S.~Rangan, T.~S. Rappaport, and
  E.~Erkip, ``Millimeter wave channel modeling and cellular capacity
  evaluation,'' \emph{IEEE J. Select. Areas Commun.}, vol.~32, no.~6, pp.
  1164--1179, Jun. 2014.

\end{thebibliography}


\end{document}